\documentclass[11pt]{article} 

\usepackage[utf8]{inputenc} 

\usepackage{geometry} 
\usepackage{authblk}

\usepackage[hidelinks]{hyperref}
\usepackage{graphicx} 
\usepackage{amsfonts,amsthm,amsmath,amssymb}
\newcommand{\R}{\mathbb{R}}
\newcommand{\C}{\mathbb{C}}

\newcommand{\I}{\mathbb{I}}
\newcommand{\N}{\mathbb{N}}
\newcommand{\ket}{\rangle}
\newcommand{\bra}{\langle}

\newcommand{\Ni}{\mathcal{N}}

\newcommand{\Bi}{\mathcal{B}}

\newcommand{\Hi}{\mathcal{H}}

\newcommand{\Mi}{\mathcal{M}}
\newcommand{\Ki}{\mathcal{K}}

\newcommand{\Di}{\mathcal{D}}
\newcommand{\Ei}{\mathcal{E}}
\newcommand{\Si}{\mathcal{S}}
\newcommand{\Fi}{\mathcal{F}}

\newtheorem{thm}{Theorem}

\newtheorem{Lemma}{Lemma}

\newtheorem*{Lemma*}{Lemma}

\newtheorem{Corollary}{Corollary}

\newtheorem*{thm*}{Theorem}

\theoremstyle{definition}
\newtheorem{Example}{Example}

\theoremstyle{definition}
\newtheorem{Definition}{Definition}

\theoremstyle{definition}

\theoremstyle{remark}
\newtheorem{Remark}{Remark}

\theoremstyle{remark}

\newtheorem*{remark*}{Remark}

\usepackage{booktabs} 
\usepackage{array} 
\usepackage{paralist} 
\usepackage{verbatim} 
\usepackage{subfig} 

\usepackage{fancyhdr} 
\usepackage{sectsty}
\allsectionsfont{\sffamily\mdseries\upshape} 

\usepackage[nottoc,notlof,notlot]{tocbibind} 
\usepackage[titles,subfigure]{tocloft} 

\title{Operator Deformations in Quantum Measurement Theory}
\author{Andreas Andersson}
\affil{\small Email: Andreas.Andersson@mis.mpg.de}
\affil{\footnotesize Institut für Theoretische Physik, Universität Leipzig. Postfach 100 920 D-04009 Leipzig, Germany\\ Max Planck Institute for Mathematics in the Sciences. Inselstrasse 22, D-04103 Leipzig, Germany}
\affil{Mathematics Subject Classification 2010 Primary: 81Q02; Secondary: 81S02, 81T02}
\affil{Keywords: Quantum Measurements, Noncommutative Quantum Field Theory, Rieffel Deformation, Warped Convolution, POVM}

\date{2013} 

\begin{document}
\maketitle
\abstract
We describe rigorous quantum measurement theory in the Heisenberg picture by applying operator deformation techniques previously used in noncommutative quantum field theory. This enables the conventional observables (represented by unbounded operators) to play a role also in the more general setting. 

\section{Introduction}
Quantum field theory, in particular quantum electrodynamics, has yielded accurate predictions not possible from other models. On the other hand, it has had a less prominent role when it comes to understanding the basic quantum features per se. In such discussions usually much simpler models are applied and it is not always clear how quantum field theory would appear in the same context. Operational quantum measurement theory \cite{BLM} has during the last two decades provided a systematic generalization of von Neumann's original formulation of measurements. Just as von Neumann's model it is one of the conceptually most important components of quantum theory since it gives an operational description of the very interactions themselves. An important task is to understand how this is related to the other parts of quantum theory, in particular quantum field theory. If this can be achieved then there is hope of understanding a great deal of physical phenomena as emerging from this type of quantum interactions.    

However, the main objects in operational measurement theory are not the observables which appear in quantum field theory. In particular, the conventional observables are represented by unbounded operators. To deal with the same formalism using conventional observables thus requires some care to ensure that the mathematical expressions make sense. The result of this paper is an identification of the measurement disturbance as a deformation which can be made mathematically rigorous and relates directly to generators of symmetry transformations and the conventional observables.  

To understand why deformations related to spacetime symmetries and observables have been studied a lot, although not as deformations due to quantum interactions, one has to appreciate the recent interest in applying algebra deformations to physics. There has been a growing interest in the idea of a ``noncommutative spacetime" in the sense of algebra deformation, i.e. the idea that the spacetime coordinates $x_0,x_1,x_2,x_3$ do not commute \cite{ADK}. While the main argument for this (related to quantum gravity) is not the motivation for this paper, investigations of noncommutativities in addition to the standard Heisenberg relation have certainly yielded a great deal of physical insight. Most directly, descriptions of many phenomena in condensed matter physics can be obtained from the case of free particles by, instead of adding interactions by hand to the Hamiltonian or Lagrangian, introducing some nonstandard commutation relations between momenta and/or coordinates (see e.g. \cite{Ho}). Again, this is interesting since the noncommutativity of quantum mechanics is due to a disturbance induced by measurement. What is then the relation between these other noncommutativities, which reproduce forces, and quantum measurements? Since the latter simply describes interactions with small systems of matter, there should exist some relation. Indeed, in a paper which will appear shortly we show using the results of this paper that the above noncommutative models can be understood from the theory of quantum measurements as initiated by von Neumann.

A measurement interaction between two quantum systems is typically modeled by a unitary operator $W=e^{-i X\otimes Y}$ acting on the composite Hilbert space $\Hi\otimes\Ki$ of the two systems. Here $X$ and $Y$ are conventional observables, i.e. unbounded selfadjoint operators in $\Hi$ and $\Ki$ respectively. If $A$ is an observable in $\Hi$ then $W^{-1}(A\otimes\I)W$ is the corresponding observable after the measurement. In Section \ref{techsection} we show that these elements are obtained as a deformation:
$$
W^{-1}(A\otimes\I)W=\int_{\R^2}\alpha_{(-y,x)}(A\otimes\I)\, dE^{X\otimes Y}(x,y).
$$
Here $E^{X\otimes Y}$ is the joint spectral measure of $X\otimes Y$ and $(x,y)\to \alpha_{(x,y)}$ is an action of $\R^2$ on suitably smooth observables, more precisely an action generated by $X$ and $Y$ (here $X$ and $Y$ can be replaced by any even integer $2N$ of generators to get an action of $\R^{2N}$). The integral above is a special case of a \textbf{warped convolution}, a mathematically sound way to obtaining noncommutative effects in physical models. It was introduced quite recently in the context of algebraic quantum field theory \cite{BS},\cite{BLS}. In addition to making the above identification, we use the warped convolution to make rigorous the formula
$$
W^{-1}(A\otimes\I)W=\int_{\R}e^{iyX}Ae^{-iyX}\otimes \, dE^Y(y),
$$
in the case when $A$ is unbounded but satisfies some requirements (Theorem \ref{maintheorem}). For this we need to extend the notion of warped convolution to the case when $A$ is not bounded. The result obtained is sufficiently general to cover the case when $A$ is a momentum or coordinate operator, or a polynomial of these (equivalently, a polynomial in annihilation and creation operators).

That the warped convolution is unitarily implemented as here makes it very peculiar. Comparing the two expression for $W^{-1}(A\otimes\I)W$ given above we see that there is a redundancy in the parameters $(y,-x)$ of $\alpha$. Nevertheless, the second expression resembles again a warped convolution but of $A$ instead of $A\otimes\I$. This will be useful when discussing physical implications in the accompanying paper.
\begin{Remark}
Warped convolution in a tensor product situation has been used also for constructing chiral quantum field theory in $1+1$ dimensions \cite{Tan} where a similar formula appeared. In that case only bounded operators were deformed and it will be covered as a special case of Theorem \ref{maintheorem} below.
\end{Remark}
After establishing this result we will in Section \ref{considerations} briefly discuss it in the language of operational quantum measurement theory. We give an explicit formula for the measurement ``instrument" and thus a recipe for how to construct such an object. These notions set the stage for an investigation of some of the basic features of quantum theory that will be left to a separate paper. We refer to that paper as ``the accompanying paper" since they are closely related. 

It  should be noted that the formalism of this paper is very general. Almost all treatments of quantum measurements assume that the algebra of observables equals  $\Bi(\Hi)$ where $\Hi$ is taken to describe the state vectors of the system of interest (as exceptions we would like to mention the work by Ojima \cite{Oj1},\cite{Oj2}, and also the very nice study \cite{Pod}; see also Section \ref{Von Neumann Measurements})\footnote{Even more commonly used are the matrix algebras $M_n(\C)$ but for most purposes we have in mind this will not be useful.}. However, suppose that the observable algebra of some quantum system is acting on the Hilbert space $\Hi_\omega$ defined by a state $\omega$ of the system, and suppose we want to consider the observables located in subregions of the system; these then usually form a proper subalgebra of $\Bi(\Hi_\omega)$ which is not isomorphic to $\Bi(\Hi)$ for any Hilbert space $\Hi$ (see e.g. \cite{Emc},\cite{Fre},\cite[sec.V.6]{Haa}). Such localized observables are interesting since interactions are localized in space and time. Also when taking the thermodynamical limit of an equilibrium system the observable algebra is usually not isomorphic to $\Bi(\Hi)$. It is therefore not without benefit that we can give a description quantum measurements which applies in a more general context. 

\section{Motivation}\label{preliminaries}
\subsection{Some Notation}
We denote by $\Hi$ or $\Ki$ the infinite-dimensional separable Hilbert space and by $\Bi(\Hi)$ the algebra of bounded operators on $\Hi$. If $X$ is a selfadjoint operator then $\operatorname{Spec}(X)$ denotes its spectrum, $E^X$ denotes its spectral measure and $E^X(\Delta)$ the corresponding projection when $\Delta\subset\R$ is a Borel set.  For two sets of operators $X_1,\dots X_N$ and $Y_1,\dots,Y_N$ in different Hilbert spaces $\Hi$ and $\Ki$ we use the notation 
$$\textbf{X}\otimes \textbf{Y}=X_\mu\otimes Y^\mu:=\sum^N_{\mu=1}X_\mu\otimes Y^\mu$$
and 
$$
e^{-i\textbf{x}\cdot\textbf{X}}:=\text{exp}\Big[-i\sum^N_{\mu=1}x_\mu X^\mu\Big],\quad\quad \textbf{x}\in\R^N.
$$
We use $\Mi_*$ to denote the set of normal linear functionals $\omega:\Mi\to\C$ of a von Neumann algebra $\Mi$. We denote by $\omega_\xi$ the vector state $\omega_\xi(A):=\bra\xi|A\xi\ket$ on $\Bi(\Hi)$ induced by an element $\xi\in\Hi$ with $\|\xi\|=1$. 

We say that a selfadjoint operator $X$ in $\Hi$ is \textbf{affiliated} to $\Mi\subset\Bi(\Hi)$ if $e^{itX}$ belongs to $\Mi$ for all $t\in\R$.

\subsection{The Measurement Process}\label{Von Neumann Measurements}
The foundations of quantum measurements were laid by von Neumann when he introduced his measurement model \cite{vN}. Generalizations in various directions have been considered, e.g. to operators with continuous spectra by Ozawa \cite{Oza}. The archetype measurement scheme concerns a quantum system with an observable $Q$ represented as an operator in a Hilbert space $\Hi$. For measuring $Q$, the interaction Hamiltonian is assumed to be of the form
$$
H_\kappa=\kappa Q\otimes \tilde{P},
$$
where $\kappa$ is some coupling constant giving the proper units and $\tilde{P}$ is an operator in another Hilbert space $\Ki$. For later comparison we note that we can spectrally decompose $Q$ and $P$ using spectral measures $E^Q$ and $E^P$ as
$$
Q= \int_{\R}x \cdot dE^Q(x), \quad\quad P=\int_{\R}p \cdot dE^P(p).
$$
The model of von Neumann assumes that the interaction takes place instantaneously (so the intrinsic evolutions of the system and apparatus do not affect the outcome). Therefore, the total (unitary) time evolution on $\Hi\otimes\Ki$  is
$$
W_\kappa:=e^{-i\kappa Q\otimes \tilde{P}}.
$$
The motivation for this choice of interaction comes from the following observation. Assume that $\tilde{P}$ has purely absolutely continuous spectra (e.g. the momentum operator) and that $\tilde{P}$ is conjugate to another operator $\tilde{Q}$, meaning that $[\tilde{Q},\tilde{P}]=i\I$. Let $\xi=\psi\otimes \phi$ be a vector in $\Bi(\Hi)\otimes\Bi(\Ki)$ defining the initial state. Suppose that $\bra\phi|\tilde{Q}\phi\ket$ is known, and for simplicity let it be $0$. Then the value of $\I\otimes \tilde{Q}$ in the final state is 
\begin{align*}
\omega_\xi(W^{-1}_\kappa(\I\otimes  \tilde{Q})W_\kappa)&=\bra e^{-i\kappa X\otimes \tilde{P}}(\psi\otimes \phi)|(\I\otimes  \tilde{Q}) e^{-i\kappa X\otimes \tilde{P}} (\psi\otimes \phi)\ket
\\&= \bra \int_{\R}d E^Q(s)\psi\otimes e^{-i\kappa s\tilde{P}}\phi|\int_{\R}d E^Q(t)\psi\otimes  \tilde{Q}e^{-i\kappa t \tilde{P}}\phi\ket 
\\&= \bra \psi|Q\psi\ket \int_{\R}\phi(q-\kappa s)q\phi(q-\kappa s)\, dq
\\&= \bra \psi|Q\psi\ket \int_{\R}\phi(q)(q+\kappa s)\phi(q)\, dq
\\&=\kappa\bra \psi|Q\psi\ket.
\end{align*}
That is, it is precisely with the unitary $W_\kappa$ that the measurement of $Q$ can be achieved. 

During the last two decades the theory of quantum measurements has been developed more systematically in the language of operational quantum theory \cite{BLM}. In this formalism ``observables" are positive operator-valued measures (POVMs).

\begin{Definition}
Let $\Omega$ be a nonempty set and let $\Fi$ be a $\sigma$-algebra of subsets of $\Omega$. A countably additive mapping $E:\Fi\to\Bi(\Hi)$ is called a \textbf{POVM} or \textbf{semispectral measure} if $0\leq E(\Delta)\leq\I$ for all $\Delta\in\Fi$ (i.e. each $E(\Delta)$ is an \textbf{effect}) and $E(\Omega)=\I$.
\end{Definition}
\begin{Definition}
A POVM $E:\Fi\to\Bi(\Hi)$ is called a projection-valued measure (\textbf{PVM}) or \textbf{spectral measure} if in addition $E(\Delta)^2=E(\Delta)$ for all $\Delta\in\Fi$ or (equivalently) $E(\Delta)E(\Delta')=0$ whenever $\Delta\cap\Delta'=\emptyset$. 

A POVM is also referred to as an \textbf{unsharp} observable while a PVMs are \textbf{sharp} observables.
\end{Definition}
Thus if we regard the spectral measure $E^X$ of a selfadjoint operator $X$ as the observable then the POVMs are ``generalized observables". There are very good reasons to argue that POVMs are needed in addition to the PVMs in order to use quantum theory in full power \cite{BLM}, many of which will be very explicit in the accompanying paper. Nevertheless, an important aspect of the tools we develop below (deformations using selfadjoint operators as generators) is that the conventional observables (e.g. multiplication and differentiation operators) can be more directly involved also in this more general formulation of measurements. 

The idea of defining a measurement process in a uniform way as below goes back to Ozawa \cite{Oza} and it has been widely used since then \cite{BLM}. We give a definition which is more general than what we could find in the literature. 

Let $\Mi\subseteq\Bi(\Hi)$ be a von Neumann algebra.
\begin{Definition}\label{meas}
A \textbf{measurement} of an observable $E:\Fi\to\Mi$ is a quintuple $(\Ki,Z,\omega_\Ki,W,f)$ where $\Ki\cong\Hi$ is the separable Hilbert space, $Z$ is a selfadjoint operator on $\Ki$, $\omega_\Ki$ is a normal state on $\Ki$, $W$ is a unitary operator on $\Hi\otimes\Ki$ (the \textbf{time evolution}) and $f:\text{Spec}(Z)\to\Omega$ is a measurable function called the \textbf{pointer function}. It is required that
\begin{equation}\label{repeat}
\omega[E(\Delta)]=(\omega\otimes\omega_\Ki)[W^{-1}(\I\otimes E^Z(f^{-1}(\Delta))W)]
\end{equation}
for all $\omega\in\Mi_*$ and $\Delta\in\Fi$.  

\end{Definition}
The meaning of Definition \ref{meas} is that the elements of $\Mi$ evolve under the measurement according to $A\to W^{-1}(A\otimes\I)W$, and similarly for those of $\Ni$, and that $\I\otimes Z$ takes the same values in the final total state as $E$ does in the initial state. With such a measurement scheme $(\Ki,Z,\omega_\Ki,W,f)$ the condition (\ref{repeat}) is usually called the \textbf{probability reproducibility condition}. When $E$ is a PVM this condition cannot hold if $E$ has continuous spectrum.  Nevertheless, the above interaction can still attempt to reproduce the values of $Q$ by means of $\tilde{Q}$ with a certain degree of inaccuracy. Replacing $Q$ by a discrete version allows perfect precision. The measured observable (see also Definition \ref{measobsdef}) is either unsharp or discrete (or both) \cite[p.119]{BLM},\cite{Oza}.

Above we took what seemed to be the most straightforward generalization of a measurement scheme for arbitrary von Neumann algebras, while restricting the time evolution to always be unitarily implementable (the approach most similar to this one can be found in \cite{Pod} where also more general evolutions are mentioned). Another way of doing this was elegantly formulated by Ojima \cite{Oj1},\cite{Oj2} where he made the brilliant identification of the above unitary $W$ with the so-called ``fundamental unitary" known from the mathematical literature of Hopf von Neumann algebras (we use the same symbol $W$ to appreciate this). 

\subsection{Rieffel Deformations and Warped Convolutions}
A few years ago, as a tool for constructing algebraic quantum field theories, e.g. for incorporating noncommutative effects of spacetime, Buchholz, Lechner and Summers \cite{BLS},\cite{BS} introduced a type of operator deformation which they called "warped convolution". The idea is as follows. Consider a set of commuting selfadjoint operators $P=(P_\mu)=(P_0,P_1,P_2,P_3)$; we could for example think of the momentum operators. These generate a 4-parameter unitary representation $x\to U(x)$ of spacetime translations in the Hilbert space defined by the physical state. There is thus an action 
$$
\alpha_x(A)=U(x)^{-1}AU(x):=e^{ixP}Ae^{-ixP}
$$
of $\R^4$ on $\Bi(\Hi)$. For a bounded operator $A$ which is smooth with respect to this action, the \textbf{warped convolution} of $A$ can be defined and equals
\begin{equation}\label{warped}
A_\Theta:=\int_{\R^4}\alpha_{\Theta x}(A)\, dE(x),
\end{equation}
where $dE(x)$ is the joint spectral measure of the $P_\mu$'s and $\Theta$ is a $4\times 4$ skew-symmetric matrix. This deformation turns out to be related to the deformed products developed by Rieffel \cite{BLS},\cite{LW}. 

The formula \eqref{warped} has interesting applications to physics, for example when the generators are the momenta $P_\mu$. Other commuting generators can be important as well. Application for \eqref{warped} has been found in excess by Albert Much \cite{Mu}. It turns out that deformations with the coordinate operators $X^\mu$ conjugate to the $P_\mu$'s actually reproduce minimal coupling to a gauge field, at least in the nonrelativistic setting (see also the accompanying paper). This is intuitive since the generators of Galilean boosts are basically the coordinate operators, and 
$$``\text{boost}\to\text{acceleration}\to\text{force}\to\text{gauge potential}."$$   
The above deformation \eqref{warped} somehow provides a path from symmetries to forces using only the commutation relations of the symmetry group. Can we also understand why this is true? 

Of concern is also the unitarity of the transformation $A\to A_\Theta$. More precisely, it is not true in general that there is a unitary operator $U\in\Bi(\Hi)$ such that an operator $A$ on $\Hi$ can be mapped to $A_\Theta$ by $A\to U^{-1}AU$. But it is known that there are situations when introducing noncommutativity by means of replacing $A$ by $A_\Theta$ can account for the difference between a system without and in the presence of an external force field. The noncommutativity should, as in the case of the Heisenberg relation, come from the interaction between two or more quantum systems. Thus, in addition to the observable algebra of the system, the other player in this interaction (which is not seen in the above description) must be included in order to obtain this unitarity. Again this will be discussed in an accompanying paper; in this paper we take care of the required underlying theory.  

\section{Deformations from Quantum Measurements}\label{techsection}
To summarize our motivation, we would like to have a relation between quantum measurements and a field description of the involved forces. In the description of quantum measurements there appears a natural tensor product $\Hi\otimes\Ki$ between the system $\Hi$ and the apparatus $\Ki$ (the latter can also be viewed as an environment; it is just any other quantum system interacting with that described by the Hilbert space $\Hi$). We consider an interaction of the form
$$
W_\kappa=e^{-i\kappa \mathbf{X}\otimes \mathbf{Y}}
$$
for some $\kappa\in\R$, where $X_1,\dots,X_N$ are affiliated to the algebra of observables $\Mi$ on $\Hi$ and $Y_1,\dots,Y_N$ are affiliated to the algebra of observables $\Ni$ on $\Ki$. As always we assume that $\Mi\subset\Bi(\Hi)$ and $\Ni\subset\Bi(\Ki)$ are von Neumann algebras. After the interaction the elements $A\otimes\I$ affiliated to $\Mi\otimes\Ni$ have in the Heisenberg picture evolved into $W_\kappa^{-1}(A\otimes\I)W_\kappa$. Consider the action $\alpha_{\mathbf{x},\mathbf{y}}(A):=U(\mathbf{x},\mathbf{y})^{-1}AU(\mathbf{x},\mathbf{y})$ of $\R^{2N}$ on $\Mi\otimes\Ni$ where
\begin{equation}\label{small action}
U(\mathbf{x},\mathbf{y}):=e^{-i\mathbf{x}\cdot \mathbf{X}}\otimes e^{i \mathbf{y}\cdot \mathbf{Y}},  \quad\quad (\mathbf{x},\mathbf{y})\in\R^{2N}.
\end{equation}
Let $\Theta$ be a skew-symmetric $2N\times 2N$ matrix. Similar to the convolution \eqref{warped} we write for operators $B$ on $\Hi\otimes\Ki$
$$
B_{\kappa\Theta}:=\int_{\text{Spec}(\mathbf{X}\otimes \mathbf{Y})}\alpha_{\kappa\Theta(\mathbf{x},\mathbf{y})}(B)\, dE^{\mathbf{X}\otimes \mathbf{Y}}(\mathbf{x},\mathbf{y}).
$$
We shall show that
\begin{equation}\label{main result}
W_\kappa^{-1}(A\otimes\I)W_\kappa=(A\otimes\I)_{\kappa\Theta}
\end{equation}
when $A$ is a reasonably well-behaved operator in $\Hi$ and $\Theta:=\bigl(\begin{smallmatrix}
0&\I\\ -\I&0
\end{smallmatrix} \bigr)$. To have the result for unbounded $A$ is necessary since all physical observables that we are going to deform are unbounded. The result will be general enough to cover all cases we are interested in. In words Equation \eqref{main result} says that:
\begin{flushleft}Post-measurement observables $e^{i\kappa \mathbf{X}\otimes \mathbf{Y}}(A\otimes\I)e^{-i\kappa \mathbf{X}\otimes \mathbf{Y}}$ in $\Hi$ are obtained from the operators $A\otimes\I$ by deformation using warped convolutions with the action $\alpha_{\mathbf{x},\mathbf{y}}=\operatorname{Ad}U(\mathbf{x},\mathbf{y})$ given in \eqref{small action} where $X$ is the operator to be measured on the quantum system. \end{flushleft}
Having this background with quantum interactions in mind, we will in this section forget about observables or von Neumann algebras and just discuss operators on Hilbert spaces in general.  

\subsection{Technical Results}
\begin{Lemma}\label{the lemma}
Let $X$ and $Y$ be arbitrary selfadjoint operators on Hilbert spaces $\Hi$ and $\Ki$ respectively. Then
\begin{equation}\label{inproof}
e^{-iX\otimes Y}=\int_{\operatorname{Spec}Y}e^{-iyX}\otimes dE^Y(y)
\end{equation}
in the weak sense. 
\end{Lemma}
\begin{proof}
Note first that the spectral measure of $X\otimes Y$ is $dE^{X\otimes Y}(x,y)=dE^X(x)\otimes dE^Y(y)$ \cite{Fox}. For any two vectors $\psi_1\otimes\phi_1$ and $\psi_2\otimes\phi_2$ in $\Hi\otimes\Ki$ we have
\begin{align*}
\bra\psi_1\otimes\phi_1|e^{-iX\otimes Y}\psi_2\otimes\phi_2\ket&=\bra\psi_1\otimes\phi_1|\Big(\int_{\R^2}dE^X(x)e^{-ixy}\otimes dE^Y(y)\Big)\psi_2\otimes\phi_2\ket
\\&:=\int_{\R^2}\bra\psi_1|dE^X(x)\psi_2\ket e^{-ixy}\bra\phi_1|dE^Y(y)\phi_2\ket.
\end{align*}
These are just ordinary (i.e. scalar-valued) integrals and the integrand is integrable in each variable separately. By Fubini's theorem we have
\begin{align*}
&\int_{\R^2}\bra\psi_1|dE^X(x)\psi_2\ket e^{-ixy}\bra\phi_1|dE^Y(y)\phi_2\ket.
\\&=\int_{\R}\bra\psi_1|e^{-iyX}\psi_2\ket \bra\phi_1|dE^Y(y)\phi_2\ket.
\\&=:\bra\psi_1\otimes\phi_1|\Big(\int_{\text{Spec}Y}e^{-iyX}\otimes dE^Y(y)\Big)\psi_2\otimes\phi_2\ket.
\end{align*}
Thus \eqref{inproof} holds in the sense of matrix elements for product vectors and therefore on all of space by continuity.  
\end{proof}

In the following theorem we denote by $\Di_A$ the domain of an operator $A$ and for a selfadjoint operator $X$ we denote by $\Di_X^\infty$ its smooth domain, i.e. the intersection of the domains of all powers of $X$. 

\begin{thm}\label{maintheorem}
Let $X_1,\dots X_N ,Y_1,\dots,Y_N$ be arbitrary commuting selfadjoint operators in Hilbert spaces $\Hi$ and $\Ki$ respectively. Consider the representation of $\R^{2N}$ in $\Hi\otimes\Ki$ given by
$$
U(\mathbf{x},\mathbf{y}):=e^{-i\mathbf{x}\cdot\mathbf{X}}\otimes e^{i\mathbf{y}\cdot\mathbf{Y}},  \quad\quad (\mathbf{x},\mathbf{y})\in\R^{2N},
$$
and suppose $A$ is an operator on $\Hi$ for which there exists a dense subspace $\Di\subset\Hi$ such that
$$
A\Di\subset\Di, \quad\quad X\Di\subset\Di, \quad\quad e^{-i\mathbf{x}\cdot\mathbf{X}}\Di\subset\Di, \quad\quad \forall \mathbf{x}\in\R^N
$$
and for all $\psi\in\Di$ there exists a constant $c_\psi$ such that
\begin{equation}\label{polbound}
\|Ae^{-i\mathbf{y}\cdot\mathbf{X}}\psi\|\leq c_\psi(1+\|\mathbf{y}\|^2)^{m/2}, \quad\quad \forall\mathbf{y}\in\R^N
\end{equation}
with $m\in\N$ independent of $\psi$. Then the following expressions are well-defined (in the weak sense) on the algebraic tensor product $\Di\odot\Di_Y^\infty$ and we have the equalities
\begin{align*}
&e^{i\mathbf{X}\otimes \mathbf{Y}}(A\otimes\I)e^{-i\mathbf{X}\otimes \mathbf{Y}}
\\&=\int_{\mathrm{Spec}\mathbf{Y}}e^{i\mathbf{y}\cdot\mathbf{X}}Ae^{-i\mathbf{y}\cdot\mathbf{X}} \otimes dE^\mathbf{Y}(\mathbf{y})
\\&=\int_{\R^{2N}} dE^{\mathbf{X}\otimes \mathbf{Y}}(\mathbf{x},\mathbf{y})U(\Theta(\mathbf{x},\mathbf{y}))^{-1}(A\otimes\I)U(\Theta(\mathbf{x},\mathbf{y}))\equiv(A\otimes\I)_\Theta
\end{align*}
where $\Theta:=\bigl(\begin{smallmatrix}
0&\I\\ -\I&0
\end{smallmatrix} \bigr)$.
\end{thm}
\begin{Remark}
Trivially, the same is true if $A\otimes\I$ is replaced by $A\otimes B$ whenever $B\in\Bi(\Ki)$ strongly commutes with $Y$. 
\end{Remark}

\begin{proof}
We consider the case of $N=1$ since it is similar for all $N$. If we can show that the warped convolution is defined for $A$ satisfying the requirements of the theorem, then
\begin{equation*}
\int_{\R\times\R} dE^{X\otimes Y}(x,y)U(-y,x)(A\otimes\I)U(-y,x)^{-1}=\int_{\R}e^{iyX}Ae^{-iyX} \otimes dE^Y(y)
\end{equation*}
holds by definition of $U(x,y)$ and Lemma \ref{the lemma}. Thus it remains to show that the above integral representation holds.
Using $e^{ixY}= \int_{\R}e^{ixk}dE^Y(k)$ we can write
$$
\int_{\R}e^{iyX}Ae^{-iyX} \otimes dE^Y(y)=\frac{1}{2\pi}\int_{\R}\int_{\R}e^{iyX}Ae^{-iyX} \otimes e^{ikY}e^{-iyk}\, dy\, dk,
$$ 
which is the oscillatory-integral form of the warped convolution in the present case. We show that the integrand of the above double integral is smooth in $y$ and $k$ in a suitable locally convex topology and that all derivatives are polynomially bounded. Then we can just rely on \cite{LW} where these properties of the integrand of the oscillatory warped convolution integral were shown to be sufficient for the validity of the same integral.  

The smoothness of $k\to e^{ikY}$ on $D^\infty_Y$ is clear so we need only consider the left tensor factor of the integrand. Define the operator-valued function
$$
F(y):=e^{iyX}Ae^{-iyX},\quad\quad y\in\R.
$$
The matrix elements $\bra \psi|e^{iyX}Ae^{-iyX}\psi\ket$ are well defined for $\psi\in\Di$ by assumption. The corresponding matrix elements of the derivatives of $F$ consist of finite linear combinations of terms of the form
$$
|\bra \psi|X^ne^{iyX}Ae^{-iyX}X^s\psi\ket|\leq \|X^n\psi\|\cdot \|Ae^{-iyX}X^s\psi\|
$$
for $n,s\in\N$. By assumption we have that $X^s\psi\in\Di$, and also that $\|X^n\psi\|$ is finite. Therefore, the condition \eqref{polbound} gives
$$
|\bra \psi|\partial^n_yF(y)\psi\ket|\leq c_{n,\psi}(1+|y|^2)^{m/2}
$$
for some finite number $ c_{n,\psi}$, where $\partial^n_y$ denotes the $n$th order derivative with respect to $y$. Therefore, all derivatives of $F$ are polynomially bounded and smooth with respect to matrix elements. 

The warped convolution with integrand $F(y)\otimes e^{ikY}$ is thus defined on the algebraic tensor product $\Di\odot\Di_Y^\infty$
of finite linear combinations of product vectors $\psi\otimes\phi$. Hence the equalities given in the theorem hold on this dense subspace of $\Hi\otimes\Ki$, and the equality with $e^{iX\otimes Y}(A\otimes\I)e^{-iX\otimes Y}$ shows that this defines an operator with domain $\Di\odot\Di_Y^\infty$. 
\end{proof}

\begin{Example}[Bounded operators]
The case when both the generators $X,Y$ and the operator $A$ to be deformed are bounded is obtained by taking $\Di=\Hi,\Di_Y^\infty=\Ki$ and $m=0$. However, in this case there is a much simpler alternative proof, which we now briefly sketch. The first step is to make a power series expansion of $e^{-iX\otimes Y}$
\begin{align*}
e^{-iX\otimes Y}&=\sum_{n=0}^\infty \frac{(-i)^n}{n!}X^n\otimes Y^n
\\&=\sum_{n=0}^\infty \frac{(-i)^n}{n!}X^n\otimes \int_{\R}y^n\, dE^Y(y)
\\&=\int_{\R}\sum_{n=0}^\infty \frac{(-i)^n}{n!}y^nX^n\otimes  dE^Y(y)
\\&=\int_{\R}e^{-iyX}\otimes  dE^Y(y),
\end{align*}
using the spectral theorem for $Y$ and Lebesgue's dominated convergence theorem in the sense of matrix elements to interchange the limits. Then for $A\in\Bi(\Hi)$ this can be directly applied to give the result. However, bounded operators never appear in our applications.
\end{Example}

\begin{Example}[Main Application]
To show that the above result is sufficiently general for physical applications, let $X$ be the position operator on $L^2$ with domain $\Di_X=\{f\in L^2(\R)|\int_\R|xf(x)|^2dx<\infty$ and let $A=P$ be the momentum operator with dense domain defined as the set of $f\in L^2(\R)$ with weak derivative in $L^2(\R)$. Then viewing Schwartz space $\Si(\R)$ as a subspace of $L^2(\R)$ we can take $\Di=\Si(\R)$ because this space is invariant under $X, P$ and $e^{-ixX}$ for all $x\in\R$ and
$$
\|Pe^{-ixX}\psi\|=\|e^{ixX}Pe^{-ixX}\psi\|=\|(P+ix)\psi\|\leq c_\psi(1+|x|^2)^{1/2}
$$
for all $n\in\N$ and $\psi\in\Si(\R)$. Thus $m=1$ suffices in this case. More generally, the powers of $P^m$ also satisfies the properties of $A$ in the theorem since
$$
\|P^me^{-ixX}\psi\|\leq c_\psi(1+|x|^2)^{m/2}, \quad\quad \forall \psi\in\Si(\R).
$$
\end{Example}

\subsection{Remarks}

\begin{Remark}[Formalism without tensor product]\label{without tensor product}
We now emphasize which properties of the tensor product were used in proving Theorem \ref{maintheorem}. At the same time we obtain a slight modification of the construction which may turn out to be useful when discussing local operators represented on the same Hilbert space in future papers. 

Let $X$ and $Y$ be commuting selfadjoint operators acting on $\Hi$ such that the spectral measure of $XY$ is the product $E^XE^Y$. Let $A$ be an operator on $\Hi$ satisfying the properties of Theorem \ref{maintheorem} and assume $[Y,A]=0$. Then it follows as in the tensor product version that
$$
e^{iXY}Ae^{-iXY}=\int_{\R}e^{iyX}Ae^{-iyX} \, dE^Y(y).
$$
The warped convolution which is used to prove that the above formula holds for unbounded $A$ is defined by the action
$$
\alpha_{(x,y)}(A):=e^{-ixX}e^{iyY}Ae^{-iyY}e^{ixX},
$$
which equals $e^{-ixX}Ae^{ixX}$ when $[Y,A]=0$.

This being said, unless in a particular representation of the operator algebras it is useful to keep the tensor product since it makes clear when the joint spectral measure is really a product. 
\end{Remark} 

\begin{Remark}[Relation to Twist Deformation]
With the standard symplectic structure as deformation matrix the warped convolution becomes a representation of the well-known Moyal $\star$-product (or rather the ``twisted convolution" related to the $\star$-product via some intertwining Fourier transforms \cite{BoH}). For suitable functions $f,g:\R^{2n}\to\R$ this product is given by a well-defined Rieffel-integral formula \cite{Po}: 
$$
(f\star g)(x)=\int_{\R^{2n}}\int_{\R^{2n}}f(x+\Theta z)g(x+y)e^{2\pi i z\cdot y}\, dz \, dy
$$
where $\Theta=\bigl(\begin{smallmatrix}0&\I\\ -\I&0\end{smallmatrix} \bigr)$ is the standard symplectic structure on $\R^n$
(indeed, this formula inspired Rieffel's generalized setting). Moreover, if we introduce a small parameter $\hbar$ and replace $\Theta$ by $\hbar\Theta$ then expanding the function $t\to e^{2\pi it}$ in a Taylor series gives \cite{EGB} (denoting the resulting product by $\star_\hbar$)
$$
(f\star_\hbar g)(x) \sim \sum_{|\alpha|=0}^\infty\frac{(i\hbar)^{|\alpha|}}{\alpha!}(-1)^{\bra\alpha\ket}\partial^{\alpha}f\partial^{\alpha}g, \quad\quad \text{when } \hbar\to 0
$$
for multi-indices $\alpha=(\alpha_1,\dots,\alpha_{2n})\in\N^{2n}$ and $|\alpha|:=\alpha_1+\cdots\alpha_{2n}, \bra\alpha\ket:=\alpha_n+\cdots+\alpha_{2n}$ (with e.g. $1,\dots,n$ the positions and $n,\dots,2n$ the momenta). This is the reason why the Moyal product is sometimes written as \cite{ADK}
$$
(f\times_{\hbar\Theta} g)(x)=m_\otimes \Big[e^{i\hbar\Theta_{jk}\partial^j\otimes \partial^k/2}(f\otimes g)(x)\Big]
$$
with $m_\otimes(f\otimes g)(x):=f(x)g(x)$. The tensor product structure can therefore be used to obtain an asymptotic expansion of the Moyal product by applying a unitary transformation. Sometimes the partial derivatives are identified with the momentum operators in the Schrödinger representation. The operator $e^{i\hbar\Theta_{jk}\partial^j\otimes \partial^k/2}$ is an example of a \textbf{twist}. Varying the skew-symmetric matrix $\Theta$ yields more general twists $e^{i\hbar\Theta_{jk}\partial^j\otimes \partial^k/2}$ which have been used in the context of Lie algebra deformations and in models of noncommutative spacetime \cite{FW},\cite[Sec.8]{ADK}. The tensor product appearing in the twist does not have the same interpretation as the measurement-motivated tensor product we discuss in this paper because it is a tensor product of the same algebra. That there is a connection between the twist and warped convolutions is well-known since they both can reproduce Moyal-Weyl space, but  how twist-looking operators could appear in the warped convolution context was unknown. Here we see that they are particular cases of our unitary transformation, although the multiplication map $m_\otimes$ makes them not really equivalent. Also, the twist deformations exist only as formal power series, in contrast to the warped convolution. Nevertheless we can give a clearer physical and mathematical meaning to at least some of these deformations. In fact, as we shall see in another paper, the measurement coupling reproducing minimal coupling to an electromagnetic gauge potential turns out to be the usual twist operator in three dimensions but with coordinate operators instead of Schrödinger-represented momentum operators. When the algebra of the spacetime coordinate is twisted (from commutative to noncommutative) using the latter unitary, the resulting spacetime is referred to as \textbf{Moyal-Weyl space}. The corresponding \emph{momentum} Moyal-Weyl space can also be obtained by doing warped convolutions; this is done in \cite{Mu}. 
\end{Remark}

\section{Some Notions from Operational Measurements}\label{considerations}
Let $A\in\Mi$ be an observable. So far we have investigated the element $W^{-1}(A\otimes\I)W$ corresponding to $A$ after an interaction $W=\text{exp}(-iX\otimes Y)$ with some other system $\Ki$. But $A\otimes\I$ is an operator on the composite system $\Hi\otimes\Ki$. The evolution of $A$ is obtained after choosing an initial state $\omega_\Ki$ on $\Bi(\Ki)$ and evaluating $W^{-1}(A\otimes\I)W$ in $\I\otimes\omega_\Ki$. This last step is similar to the partial trace operation on states (but it is not its dual).
 
If $\omega_\Ki\in\Bi(\Hi)_*$ is the initial state on $\Ki$ then the time evolution of an element $A\in\Mi$ is given by 
$$A\to\int_{\text{Spec}Y}e^{iyX}Ae^{-iyX} \omega_\Ki[dE^Y(y)].$$
Now it may be that the outcome of the measurement is recorded by measuring the pointer observable $E^Z$ conjugate to $E^Y$. In that case the evolution of $A$ can be made more precise; it is zoomed in using the outcome of the measurement. For this we use the notion of ``instruments" \cite{DL}. Namely, for all Borel subsets $\Delta$ of $\R$ we have the map
$$
\Ei^*_{\Delta}:\Mi\to\Mi, \quad\quad A\to (\I\otimes\omega_\Ki)[W^{-1}(A\otimes E^Z(f^{-1}(\Delta)))W],
$$
which defines the \textbf{dual} $\Ei^*:\Delta\to\Ei^*_\Delta$ of the instrument of the interaction described by $X\otimes Y$ (here $f:\text{Spec}(Z)\to\text{Spec}(X)$ is a function relating the spectra as in Definition \ref{meas}). 
\begin{Remark}
That each $\Ei_\Delta^*$ takes values in the algebra $\Mi$ is guaranteed since $X$ was assumed to be affiliated to $\Mi$.
\end{Remark}
The \textbf{instrument} $\Delta\to(\Ei_\Delta:\Mi_*\to\Mi_*)$ is then defined via
$$
(\Ei_{\Delta}(\rho))(A)=\rho(\Ei^*_{\Delta}(A)), \quad\quad A\in\Mi.
$$
Thus, the map $\Ei_{\Delta}$ on states corresponds to the Schrödinger picture while the map $\Ei_{\Delta}^* $ on observables corresponds to the Heisenberg picture.
For our purposes however, the most important use of instrument is that it can show us how the statistics of the measured observable are generally not the same as one might have guessed. 
\begin{Definition}\label{measobsdef}
Let $\Ei:\Mi_*\to\Mi_*$ be an instrument. The \textbf{measured observable} $E^{\Ei}$  associated to $\Ei$ is defined by
$$
E^{\Ei}(\Delta):=\Ei_\Delta^*(\I).
$$
\end{Definition}
\begin{Remark}
For a given instrument $\Ei$ the measured observable $E^\Ei$ is unique. On the other hand, there are many instruments which define the same observable $E^\Ei$.  
\end{Remark}
Since we have obtained an explicit formula for the deformed observables, the observable associated to $\Ei$ can be calculated explicitly.
\begin{Corollary}
Let $(\Ki,Z,\omega_\Ki,e^{-i\kappa X\otimes Y},f)$ be a measurement of an observable $X$ as in Definition \ref{meas} and assume that $[Y,Z]=i\I$. Then the measured observable is given by  
\begin{equation*}\label{measobs}
E^{\Ei}_\kappa(\Delta)=\int_{\mathrm{Spec}X}dE^X(x)\omega_\Ki[dE^Z(f^{-1}(\Delta-\kappa x))]. 
\end{equation*}
\end{Corollary}
\begin{proof}
Let $\Ei$ be the corresponding instrument. The result follows since we assumed $E^Z$ to be conjugate to $E^Y$:
\begin{align*}
\Ei_\Delta^*(\I)&=\int_{\mathrm{Spec}X}dE^X(x) \omega_\Ni[e^{i\kappa x Y}E^Z(f^{-1}(\Delta))e^{-i\kappa xY}] 
\\&=\int_{\mathrm{Spec}X}dE^X(x)\omega_\Ni[dE^Z(f^{-1}(\Delta-\kappa x))]. 
\end{align*}
\end{proof}
Therefore, unless $\omega_\Ki[dE^Z(f^{-1}(\Delta-\kappa x))]=\chi_\Delta$ for all $x\in\text{Spec}(X)$, the measured observable will not be equal to the one ``intended" to be measured, i.e. the one appearing in the interaction Hamiltonian. When we have $[Y,Z]=i\I$, this can only happen if $\kappa=0$. On the other hand, if $[Y,Z]=0$ then $\Ei_\Delta^*(\I)=\omega_\Ki[E^Z(f^{-1}(\Delta))]\I$ for all $\Delta$ so that only a multiple of the identity can be measured (which is usually far from $E^X$!). This manifests the trade-off between accuracy and disturbance. 

\section{Acknowledgements}
The author thanks his college Albert Much and supervisors Gandalf Lechner and Rainer Verch for discussions and proofreading.

\section{References}
\bibitem[ADK]{ADK} Aschieri P, Dimitrijevi\'c M, Kulish P, Lizzi F, Wess J. Noncommutative spacetimes: symmetries in noncommutative
geometry and field theory. Lect. Notes Phys. 774  (2009).

\bibitem[BoH]{BoH} Bowes D, Hannabuss KC. Weyl quantization and star products. Journal of Geometry and Physics 22, 319-348 (1997).

\bibitem[BLS]{BLS} Buchholz D, Lechner G, Summers S. Warped convolutions, Rieffel deformations and the construction of quantum field theories. Commun. Math. Phys. 304, 95-123 (2011). 

\bibitem[BS]{BS} Buchholz D, Summers S. Warped convolutions: a novel tool in the construction of quantum field theories. In: Quantum field theory and beyond. World Scientific, Singapore, pp. 107–121. arXiv:0806.0349v1 (2008). 

\bibitem[BLM]{BLM} Busch P, Lahti PJ,  Mittelstaedt P. The quantum theory of measurement – Second revised edition.
Springer-Verlag, Berlin (1996).

\bibitem[DL]{DL} Davies EB, Lewis JT. An operational approach to quantum probability. Comm. Math. Phys. 18, 239 (1970).

\bibitem[Emc]{Emc} Emch GC. Algebraic methods in statistical mechanics and QFT. Wiley. NewYork (1972).

\bibitem[EGB]{EGB} Estrada R, Garcia-Bond\'{i}a, V\'{a}rilly JC. On asymptotic expansions of twisted products. J. Math. Phys. 30, 2789 (1989).

\bibitem[FW]{FW} Fiore G, Wess J. On “full” twisted Poincar\'{e} symmetry and QFT on Moyal-Weyl spaces. arXiv:hep-th/0701078v3 (2007).

\bibitem[Fox]{Fox} Fox DW. Spectral measures and separation of variables. Journal of Research of the National Bureau of Standards. Section B: Mathematical Sciences Vol. 80B, No.3, 347-351 (1976).

\bibitem[Fre]{Fre} Fredenhagen K. On the modular structure of local algebras of observables. Commun. Math. Phys. 97,79-89 (1985).

\bibitem[Haa]{Haa} Haag R. Local quantum physics. Springer-Verlag, Berlin (1992).

\bibitem[HO]{HO} Harada R, Ojima I. A unified scheme of measurement and amplification processes based on micro-macro duality – Stern-Gerlach experiment as a typical example – arXiv:0810.3400v1 (2008).

\bibitem[Ho]{Ho} Horv\'{a}ty PA. Non-commutative mechanics in mathematical and in condensed matter physics. arXiv:cond-mat/0609571v2. Symmetry, Integrability and Geometry: Methods and Applications Vol. 2 (2006).

\bibitem[LW]{LW} Lechner G, Waldmann S. Strict deformation quantization of locally convex algebras and modules. arXiv:1109.5950 (2011).

\bibitem[Mu]{Mu} Much A. Quantum mechanical effects from deformation theory. arXiv:1307.2609 (2013).

\bibitem[Oj1]{Oj1} Ojima I. Micro-macro duality in quantum physics arXiv:math-ph/0502038v1 (2005).

\bibitem[Oj2]{Oj2} Ojima I. Micro-macro duality and emergence of macroscopic levels. Talk at the International Symposium, QBIC (2007).

\bibitem[Oza]{Oza} Ozawa M.  Quantum measuring processes of continuous observables. J. Math. Phys. 25, 79 (1984).

\bibitem[Pod]{Pod} Podsedkowska H. Correlations in a general theory of quantum measurement. Open Sys. and Information Dyn.14:445–458 (2007).

\bibitem[Po]{Po} Pool JCT. Mathematical aspects of the Weyl correspondence. J. Math. Phys. 7, 66; doi: 10.1063/1.1704817 (1966).

\bibitem[Rie]{Rie} Rieffel MA. Deformation quantization for actions of $\R^d$. Mem. Amer. Math. Soc., 106(506), (1993).

\bibitem[Tan]{Tan} Tanimoto Y. Construction of wedge-local nets of observables through Longo-Witten endomorphisms. arXiv:1107.2629v2 (2012).

\bibitem[vN]{vN} Von Neumann J. Mathematical foundations of quantum mechanics. Princeton University Press, Princeton (1955).

\end{document}